\documentclass{article}
\usepackage[utf8]{inputenc}
\usepackage[T1]{fontenc}
\usepackage[letterpaper, margin=2cm]{geometry}
\usepackage{amsmath, amsthm, amsfonts, amssymb}
\usepackage{halloweenmath} 
\usepackage{wasysym} 
\usepackage{epigraph}
\usepackage{enumitem}
\usepackage{accents} 
\usepackage[noend]{algpseudocode}
\usepackage{float} 
\usepackage[dvipsnames]{xcolor} 
\usepackage{tikz}
\usepackage[framemethod=tikz]{mdframed}
\usepackage{hyperref}

\usetikzlibrary{patterns} 

\hypersetup{
    colorlinks,
    linkcolor={red!50!black},
    citecolor={blue!50!black},
    urlcolor={blue!80!black}
}

\usepackage{xpatch}
\makeatletter
\xpatchcmd{\endmdframed}
  {\aftergroup\endmdf@trivlist\color@endgroup}
  {\endmdf@trivlist\color@endgroup\@doendpe}
  {}{}
\makeatother

\makeatletter
\g@addto@macro\bfseries{\boldmath}
\makeatother

\newtheoremstyle{customstyle}
  {}
  {}
  {}
  {}
  {\bfseries}
  {.}
  { }
  {\thmname{#1}\thmnumber{ #2}\thmnote{ (#3)}}%
\theoremstyle{customstyle}

\newmdtheoremenv[innertopmargin=-2pt, skipabove=5pt, skipbelow=0pt, linewidth=2pt, linecolor=Green, backgroundcolor=Green!10, bottomline=false, leftline=false, rightline=false]{theorem}{Theorem}
\newmdtheoremenv[innertopmargin=-2pt, skipabove=5pt, skipbelow=0pt, linewidth=2pt, linecolor=Green, backgroundcolor=Green!10, bottomline=false, leftline=false, rightline=false]{lemma}{Lemma}

\newcommand{\F}{\ensuremath{\mathcal{F}}}
\newcommand{\iGoedel}{\ensuremath{\accentset{\scalebox{0.35}{$\infty$}}{G}}}
\newcommand{\iRosser}{\ensuremath{\accentset{\scalebox{0.35}{$\infty$}}{R}}}

\setlist[description]{leftmargin=0.5cm, style=nextline, noitemsep, topsep=0.25pt}
\setlist[itemize]{style=nextline, itemsep=0.25pt, topsep=0.25pt, label={$\Rightarrow$}}

\newenvironment{algo}{\begin{samepage}\medskip\hrule\begin{algorithmic}}{\end{algorithmic}\hrule\medskip\end{samepage}}

\begin{document}

\title{\vspace{-1cm}Incompleteness Ex Machina}
\author{Sebastian Oberhoff\\{\small oberhoff.sebastian@gmail.com}}
\date{\today}

\maketitle

\begin{abstract}
In this essay we'll prove Gödel's incompleteness theorems twice. First, we'll prove them the good old-fashioned way. Then we'll repeat the feat in the setting of computation. In the process we'll discover that Gödel's work, rightly viewed, needs to be split into two parts: the transport of computation into the arena of arithmetic on the one hand and the actual incompleteness theorems on the other. After we're done there will be cake.
\end{abstract}

\epigraph{It is a profoundly erroneous truism, repeated by all copy-books and by eminent people when they are making speeches, that we should cultivate the habit of thinking of what we are doing. The precise opposite is the case. Civilization advances by extending the number of important operations which we can perform without thinking about them. Operations of thought are like cavalry charges in a battle---they are strictly limited in number, they require fresh horses, and must only be made at decisive moments. }{\textit{Alfred North Whitehead}}

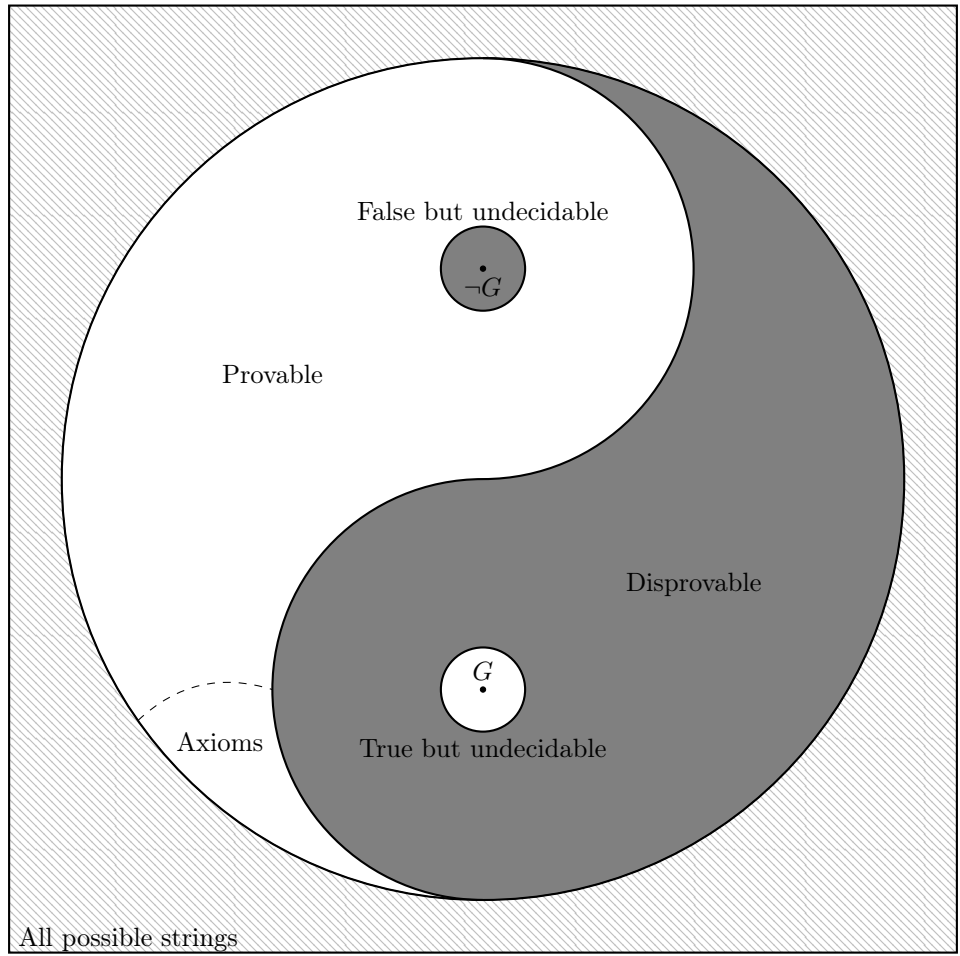
\vspace{-5cm}
\begin{figure}[H]
\centering
\begin{tikzpicture}[scale=0.7]
  \draw[pattern=north west lines, pattern color=lightgray, thick] (-9cm, -9cm) rectangle (9cm, 9cm);

  \begin{scope}
    \clip (0cm, -8cm) rectangle (8cm, 16cm);
    \draw[fill=gray] (0cm, 0cm) circle (8cm);
  \end{scope}
  \begin{scope}
    \clip (0cm, -8cm) rectangle (-8cm, 16cm);
    \draw[fill=white] (0cm, 0cm) circle (8cm);
  \end{scope}
  
  \fill[gray] (0cm, -4cm) circle (4cm);
  \fill[white] (0cm, 4cm) circle (4cm);

  \begin{scope}
    \clip (0cm, -8cm) rectangle (-8cm, 16cm);
    \draw[thick] (0cm, -4cm) circle (4cm);
  \end{scope}
  \begin{scope}
    \clip (0cm, -8cm) rectangle (8cm, 16cm);
    \draw[thick] (0cm, 4cm) circle (4cm);
  \end{scope}
  
  \draw[thick] (0cm, 0cm) circle (8cm);

  \draw[fill=white, thick] (0,-4) circle (0.8cm);
  \draw[fill=gray, thick] (0,4) circle (0.8cm);
  
  \node[right] at (-9cm, -8.75cm) {All possible strings};
  
  \draw[dashed, bend left] (215: 8cm) to (-4cm, -4cm);
  \node at (-5cm, -5cm) {Axioms};
  \node at (-4cm, 2cm) {Provable};
  \node at (0cm, -5.1cm) {True but undecidable};
  \draw[fill] (0cm, -4cm) circle (0.05cm) node[above] {$G$};
  
  \node at (4cm, -2cm) {Disprovable};
  \node at (0cm, 5.1cm) {False but undecidable};
  \draw[fill] (0cm, 4cm) circle (0.05cm) node[below] {$\neg G$};
\end{tikzpicture}
\end{figure}

\section{The Incompleteness Theorems On Fast-Forward}

Kurt Gödel's incompleteness theorems are clearly the most significant results in the history of mathematics (fight me). The first establishes that no single ``proper'' formal system can fully settle all mathematical questions; that truth and provability are distinct concepts.\footnote{Note to the veteran reader: this entire presentation is focused on the standard interpretation of arithmetic. No funky transfinite numbers.} The second shows that such a formal system also can't prove itself free of contradiction. If it tries, its wings will melt and it will crash to the ground.

Gödel accomplished these triumphs as follows. Suppose we're given a formal system\footnote{For the purpose of this discussion every formal system is effectively axiomatized by definition. This basically just boils down to the fact that proofs are computer checkable.} $\F$ that is capable of reasoning about elementary arithmetic. Then, as Gödel showed through a lot of toil~\cite{goedel}, it is possible to construct a sentence $G$ which essentially says ``I am not provable in $\F$''. Once he had built this sentence he then simply asked: can $G$ be proven or disproven---that is \textit{decided}---in $\F$?

\begin{description}
\item[Suppose $G$ is provable:]
\begin{description}
\item[]
\item[Either $\neg G$ is also provable:]
\begin{itemize}
\item[]
\item Both $G$ and $\neg G$ are provable in $\F$. Any system in which such a situation arises is also called \textit{inconsistent}.
\end{itemize}
\item[Or $\neg G$ isn't provable:]
\begin{itemize}
\item[]
\item Since $G$ is provable there exists some concrete proof of $G$ which we can write out \textit{inside} $\F$. This is done using a device called Gödel numbering which allows numbers to talk about proofs. It's complicated.
\item Because this proof leads to $G$ this allows us to construct a new proof of <<$G$ is provable>> = $\neg G$.\footnote{Whereas regular quotes (`` '') perform their usual function, guillemets (<< >>) surround formal sentences denoted by their English description. Also ``='' doesn't always have to mean exact equality. I also use it to relate sentences which are merely logically equivalent, provided this is obvious.} But we're assuming $\neg G$ \textit{isn't} provable. We have arrived at a genuine contradiction, not just an inconsistency. So this can't happen. We'll use the shorthand ``\lightning'' for this in the future.
\end{itemize}
\end{description}
\item[Suppose $\neg G$ is provable:]
This is what it means to disprove $G$. But beware: $\neg G$ isn't a self-referential sentence. In particular, it's not equal to <<$\neg G$ isn't provable>>. Instead, it's <<$G$ \textit{is} provable>>. These are very different sentences. $\neg G$ is really more like an evil twin that's telling lies about its sibling (assuming the system $\F$  as a whole is honest). Furthermore, mind the fact that this case and the previous aren't logical opposites. That's where I use either/or.
\begin{description}
\item[Either $G$ is also provable:]
\begin{itemize}
\item[]
\item $\F$ is inconsistent.
\end{itemize}
\item[Or $G$ isn't provable:]
\begin{itemize}
\item[]
\item Since $\neg G$ is provable there exists some concrete proof of $\neg G$ which we can write out \textit{inside} $\F$.
\item But in order to force a contradiction we'll have to find a proof of $G$ = <<$G$ isn't provable>>, not <<$\neg G$ is provable>> (uh-oh).
\item We can't find a proof of $G$ despite the fact that its existence is exactly what $\neg G$ = <<$G$ is provable>> asserts.
\item We haven't convicted $\F$ of an outright inconsistency. But something is still very, very wrong with it.
\end{itemize}
\end{description}
\end{description}

These deductions, as well as the ones in the proofs to follow, are deliberately written in such a painstaking style in order to make the similarities between the first and second set of proofs as plain as possible. These are the decisive moments of our battle.

Now, it would be nice to say that we've proven that, if our formal system $\F$ is consistent, then $G$ is undecidable. Unfortunately our second case wasn't strong enough. So, just as Gödel did in 1931, we'll have to settle for a weaker formulation:

\begin{theorem}[First Incompleteness Theorem---Original Version]
Let $\F$ be an \textit{honest}\footnotemark\ formal system capable of reasoning about elementary arithmetic. Then $\F$ is incomplete; it contains a sentence that can neither be proven nor disproven in $\F$.
\end{theorem}
\footnotetext{The traditional term is ``\textit{sound}''. But when explaining this concept to my mother I had to admit that ``honesty'' is much more fitting. Also, I've been burned by overloaded usage of the term ``soundness'' in the past.}

\begin{proof}\vspace{-\topsep}
We just did that.
\end{proof}

Okay, I pulled a fast one there. Crucially, I skipped over a rather important concept: \textit{honesty}. What \textit{exactly} do I mean by that?

As the name suggests, an honest formal system is a formal system that speaks the truth. We can believe the claims it makes. In particular, if our $\F$ is honest and proves $\neg G$ = <<$G$ is provable>>, then, just as $\neg G$ says, it really is the case that $G$ is provable. This is how we can force a contradiction out of the second case.

But honesty is a stronger assumption than consistency. An honest formal system is always consistent. After all, an inconsistent formal system contains a sentence $S$ such that both $S$ and $\neg S$ is provable. One of these must be a lie. However, a consistent formal system isn't necessarily honest. Here's a completely hypothetical example to illustrate: 

Imagine I'm the leader of an idyllic communist utopia. Everything is going according to current five-year plan. Except one day a nuclear reactor of mine suddenly catches fire. If I then go on to tell the world community that there's nothing to worry about, this will merely be a lie, not a contradiction. Even once Swedish scientists begin smelling the fire my excuses, while fabricated, can still retain their \textit{internal} consistency. It's only when I finally confess to my mishap that I contradict myself. At that point I have given two directly contradictory accounts. Now even those without a Geiger counter can tell that I've been dishonest.

Disorienting as it may be, this dishonest-yet-consistent state can also occur in formal systems. Most important for our purposes, it is conceivable that $\F$ can prove $\neg G =$ <<$G$ is provable>> while $G$ in fact \textit{isn't} provable. This would lead to a curious situation I'll call \textit{sub-inconsistency}\footnote{Also known under the name ``\textit{$\omega$-inconsistency}''. Let's spare our horses.} where for all $n\in\mathbb{N}$ the sentence 
\begin{center}
<<$G$ doesn't have a proof shorter than $n$>>
\end{center}
can be proven in $\F$ by just going through all possible proofs up to length $n$ and pointing out that none of them do the trick. Yet the sentence with the ``for all'' on the inside,
\begin{center}
<<$G$ doesn't have a proof shorter than $n$ for all $n\in\mathbb{N}$>> = $G$,
\end{center}
is not only unprovable but flat out contradicted by $\neg G$. This is why Gödel had to assume honesty (or at least \textit{super-consistency}) in his original proof. Basically, he had to press $\F$ up against the wall and shout: ``No games, dammit!''

Okay, but isn't there a less violent way to prevent this kind of sub-terfuge? Indeed there is. In 1936 J. Barkley Rosser found a way~\cite{rosser} to demonstrate the improved

\begin{theorem}[First Incompleteness Theorem---Rosser's Version]
Let $\F$ be a \textit{consistent} formal system capable of reasoning about elementary arithmetic. Then $\F$ is incomplete.
\end{theorem}

\begin{proof}\vspace{-\topsep}
In hindsight Rosser's trick is quite simple. Instead of formalizing ``I am not provable'' he formalized the sentence ``For every proof of me there exists a shorter disproof''. Let's refer to this sentence as $R$ and see how $R$ can help us.
\begin{description}
\item[Suppose $R$ is provable:]
$\Rightarrow$ There exists some shortest proof $p$ of $R$.
\begin{description}
\item[Either there is a proof of $\neg R$ that's shorter than $p$:]
\begin{itemize}
\item[]
\item $\F$ is inconsistent. \lightning\footnote{Inconsistencies can be promoted to genuine contradictions if they collide with the assumption of consistency.}
\end{itemize}
\item[Or there isn't:]
\begin{itemize}
\item[]
\item We can go through all strings shorter than $p$ and determine that none of them are a proof of $\neg R$.
\item We can write out this list \textit{inside} $\F$ and point out that all longer strings are at least as long as $p$.
\item This allows us to construct a new (longer) proof of <<There exists a proof of $R$ with no shorter disproof (namely $p$)>> = $\neg R$.
\item $\F$ is inconsistent. \lightning
\end{itemize}
\end{description}
\item[Suppose $\neg R$ is provable:]
$\Rightarrow$ There exists some shortest proof $p$ of $\neg R$.
\begin{description}
\item[Either there is a proof of $R$ that's shorter than $p$:]
\begin{itemize}
\item[]
\item $\F$ is inconsistent. \lightning
\end{itemize}
\item[Or there isn't:]
\begin{itemize}
\item[]
\item We can go through all strings shorter than $p$ and determine that none of them are a proof of $R$.
\item We can write out this list \textit{inside} $\F$ and point out that all longer strings are at least as long as $p$.
\item This allows us to construct a new (longer) proof of <<For every proof of $R$ there exists a shorter disproof (namely $p$)>> = $R$.
\item $\F$ is inconsistent. \lightning\qedhere
\end{itemize}
\end{description}
\end{description}
\end{proof}

Whereas Gödel's original proof was limping on one leg, Rosser's version is perfectly balanced, as all things should be, yielding the stronger result.

Finally, there's one more summit to conquer in this mathematical Himalaya. That's the Second Incompleteness Theorem.

\begin{theorem}[Second Incompleteness Theorem]
Again, let $\F$ be a consistent formal system capable of reasoning about elementary arithmetic. Then $\F$ can't prove its own consistency.
\end{theorem}

Perhaps the fact that this theorem speaks of consistency, not honesty, makes you a little suspicious. Is this another one of Rosser's upgrades? No. For the Second Incompleteness Theorem Gödel only needed half of his First Incompleteness Theorem. And, as luck would have it, this was exactly the half that spoke of consistency.

\begin{proof}
Let's assume for the purpose of contradiction that $\F$ can prove its own consistency
\begin{itemize}[topsep=0.5em]
\item The first half of the proof for the First Incompleteness Theorem can be read as a proof by contradiction that if $\F$ is consistent, then $G$ isn't provable.
\item Using the fact that <<$G$ isn't provable>> = $G$ this can be stated more concisely as ``$\F$ is consistent $\implies G$''.
\item As Gödel himself pointed out with some hand-waving, there wasn't actually anything in the reasoning we used to prove this that we couldn't carry out just as well in $\F$. So <<$\F$ is consistent $\implies G$>> is actually a theorem in $\F$.
\item We're assuming that $\F$ can prove <<$\F$ is consistent>>. So by modus ponens\footnote{Modus ponens is the inference rule that from $A$ and <<$A \implies B$>> follows $B$.} $\F$ can also prove $G$.
\item But $G$ can't be proven. \lightning\qedhere
\end{itemize}
\end{proof}

Moral of the story: if you declare yourself a very stable genius, you're not~\cite{trump}.

\section{The Rise Of The Machines}

We now come to the dramatic turn in our story. Namely, we're about to demand that our formal system $\F$ can reason about algorithms, not just elementary arithmetic. And then we'll reprove all three of the previous results on top of that.

What could possibly provoke such heresy? Well, you may recall that in the prequel we relied quite heavily on the fact that both $G$ and $R$ could \textit{somehow} be constructed in $\F$. But we never actually performed these constructions. We also liberally made assertions of the form: ``we can list out all proofs up to length $n$ \textit{inside} $\F$''. But $\F$ is about numbers, not formal sentences. How is all of this ultimately accomplished?

The reality is that I hand-waved these details for a reason. They are non-trivial in the extreme (ideal exercises for the reader!). That's because, from a modern perspective, Gödel basically had to teach programming to formal systems that were about arithmetic. And he had to do it half a decade before the first models of computation were even devised. This is where all the tough technical work had to be done. The incompleteness theorems themselves were mere victory laps at the end.

This delimitation between the different parts of Gödel's work is frequently passed over without comment but will become even more apparent as we move on. In fact, setting the record straight on this matter is really the main point of this essay. By placing computation---for example in the form of Turing machines---directly into the bedrock of our formal system our task will become orders of magnitude simpler. Gödel had to scale a vertical cliff to convince his peers that ``$p$ is a proof of $S$'' could be expressed within the confines of arithmetic. We, on the other hand, aided by the ski lift built by Turing and his apostles, can see at a glance that the same could be checked by a computer. Let's exploit that.

Anybody who still harbors a nostalgic longing for ``elementary arithmetic'' can then take on the separate task of showing how to ponder computation in the realm of the natural numbers. Perhaps one can also study computation starting from a different base camp such as knot theory. Who knows? (I'm completely clueless regarding knot theory.) Meanwhile, we will have the incompleteness theorems above the clouds all to ourselves.

\section{To Halt Or Not To Halt}

Alright, here we go. The only ingredient we need to prepare is the unsolvability of the Halting Problem due to Alan Turing in 1936~\cite{turing}.

\begin{lemma}[Unsolvability Of The Halting Problem]
There's no algorithm that implements the function
\[
A \mapsto
\begin{cases}
\text{accept} & \text{if $A(A)$ halts}\\
\text{reject} & \text{if $A(A)$ doesn't halt.}
\end{cases}
\]
(``$A(A)$'' denotes an algorithm $A$ running on its own source code.)
\end{lemma}

\begin{proof}\vspace{-\topsep}
Suppose $H$ solves the Halting Problem. Then one could create the following algorithm:
\begin{algo}
\Function{$\overline{H}(A)$}{}
  \If{$H(A)$ rejects}
    \State halt
  \EndIf
  \If{$H(A)$ accepts}
    \State loop
  \EndIf
\EndFunction
\end{algo}
``halt'' just\dots\ halts. We don't even need to bother accepting or rejecting. It doesn't matter. And ``loop'' is a simple infinite loop without exits. Now, question: does $\overline{H}(\overline{H})$ halt?
\newline
\begin{minipage}{0.5\textwidth}
\smallskip
\begin{description}
\item[Either $\overline{H}(\overline{H})$ halts:]
\begin{itemize}[noitemsep]
\item[]
\item $H(\overline{H})$ accepts.
\item $\overline{H}(\overline{H})$ doesn't halt. \lightning
\end{itemize}
\end{description}
\smallskip
\end{minipage}
\begin{minipage}{0.45\textwidth}
\smallskip
\begin{description}
\item[Or $\overline{H}(\overline{H})$ doesn't halt:]
\begin{itemize}[noitemsep]
\item[]
\item $H(\overline{H})$ rejects.
\item $\overline{H}(\overline{H})$ halts. \lightning
\end{itemize}
\end{description}
\smallskip
\end{minipage}
\newline
Either way we get a contradiction.
\end{proof}

Easy peasy. As you'd expect of a mere lemma. We're now ready for the First Incompleteness Theorem.

\begin{theorem}[First Incompleteness Theorem---Original Version By Computation]
\label{original-first-by-computation}
Let $\F$ be an honest Turing complete formal system. Then $\F$ is incomplete.
\label{first-by-comp}
\end{theorem}

For a formal system to be \textit{Turing complete} simply means that it can reason about algorithms. And our arguments won't be overly technical anyway. So let me just leave it there.

\begin{proof}
Suppose $\F$ was complete. Then we could solve the Halting Problem using the following algorithm:
\begin{algo}
\Function{$H(A)$}{}
  \For{$x\in$ all possible strings\footnote{In case it wasn't clear: we're iterating in lexicographical order. This procedure is sometimes also jokingly referred to as the \textit{British Museum algorithm} because  it's akin to using chimps in front of typewriters in an attempt to reproduce all the books in the British Museum.}}
    \If{$x$ proves <<$A(A)$  doesn't halt>>}
      \State reject
    \EndIf
    \If{$x$ proves <<$A(A)$ halts>>}
      \State accept
    \EndIf
  \EndFor
\EndFunction
\end{algo}
Because $\F$ is complete this search will eventually hit upon a proof of  either <<$A(A)$ halts>> or <<$A(A)$ doesn't halt>>. And because $\F$ is honest we'll be able to trust its judgment. That determines whether or not $A(A)$ halts. But, as we saw only a moment ago, it's impossible to solve the Halting Problem. So $\F$ couldn't have been complete. \lightning
\end{proof}

That certainly went by a lot faster than my last reading of \textit{Gödel, Escher, Bach}~\cite{hofstadter}. Though, if you've ever had the incompleteness theorems explained to you by a computer scientist, then you probably saw this proof coming a mile away. And for good reason. It's a very elegant little proof; deserving of its popularity.

\section{Welcome To The World Of Mirrors}

Nevertheless, this isn't the proof I want to use going forward. The reasons are threefold:
\begin{itemize}[topsep=0.5em, label=$\bullet$]
\item The jump from honesty to consistency is rather tricky.
\item The Second Incompleteness Theorem doesn't follow naturally at all.
\item The parallels to Gödel's and Rosser's original proofs are lost. (This one's the biggie.)
\end{itemize}
We can do better. Watch this:
\begin{proof}[Second Proof Of Theorem~\ref{first-by-comp}]
Consider the following algorithm:
\begin{algo}
\Function{$P(A)$}{}
  \For{$x\in$ all possible strings}
    \If{$x$ proves <<$A(A)$ doesn't halt>>}
      \State halt
    \EndIf
  \EndFor
\EndFunction
\end{algo}
I now claim that <<$P(P)$ doesn't halt>> is undecidable in $\F$. Let's call this sentence $\iGoedel$ because it's the perfect analogue of Gödel's $G$. The critical new feature of $\iGoedel$ is that its construction has become just a small matter of programming. If we go through the alternatives, we find:
\begin{description}
\item[Suppose $\iGoedel$ is provable:]
\begin{description}
\item[]
\item[Either $\neg \iGoedel$ is also provable:]
\begin{itemize}
\item[]
\item $\F$ is inconsistent. \lightning
\end{itemize}
\item[Or $\neg \iGoedel$ isn't provable:]
\begin{itemize}
\item[]
\item Since $\iGoedel$ is provable $P(P)$ will eventually find such a proof after some sequence of steps which we can write out inside $\F$.
\item Because this sequence leads to a terminal state this allows us to construct a new proof of <<$P(P)$ halts>> = $\neg \iGoedel$. \lightning
\end{itemize}
\end{description}
\item[Suppose $\neg \iGoedel$ is provable:]
\begin{description}
\item[]
\item[Either $\iGoedel$ is also provable:]
\begin{itemize}
\item[]
\item $\F$ is inconsistent. \lightning
\end{itemize}
\item[Or $\iGoedel$ isn't provable:]
\begin{itemize}
\item[]
\item Since $\neg \iGoedel$ is provable $P(P)$ will eventually find such a proof after some sequence of steps which we can write out inside $\F$.
\item But in order to halt $P(P)$ will have to find a proof of $\iGoedel$ = <<$P(P)$ doesn't halt>>, not $\neg \iGoedel$ (uh-oh).
\item $P(P)$ doesn't halt despite the fact that this is exactly what $\neg \iGoedel$ = <<$P(P)$ halts>> asserts.
\item $\F$ is dishonest. \lightning\qedhere
\end{itemize}
\end{description}
\end{description}
\end{proof}

Is anybody else experiencing déjà vu? Note that, whereas in the previous proof we only established the \textit{existence} of undecidable sentences in $\F$, here we actually have a concrete undecidable sentence on our hands. And we didn't even need the Halting Problem!

Moreover, almost miraculously, the two cases for $\iGoedel$'s provability run into the exact same problems as $G$. In the first case $\F$ plants its $\F$ace squarely in inconsistency. In the second one the precise correlation between the information $\F$ communicated and the facts insofar as they can be determined and demonstrated is such as to cause epistemological problems of sufficient magnitude as to lay upon the logical and semantic resources of the English language a heavier burden than they can reasonably be expected to bear. It told a lie~\cite{yes-prime-minister}.

Upon closer inspection we notice that we're even dealing with the same kind of lie:
\begin{center}
$\mathghost$ sub-inconsistency $\mathghost$
\end{center}
After all, if $\iGoedel$ isn't provable, then one can still prove that $P(P)$ hasn't halted after $n$ steps for all $n\in\mathbb{N}$ by just writing out $n$ steps of the computation in $\F$ and noting that it's still running. But the summarizing sentence <<$P(P)$ hasn't halted after $n$ steps for all $n\in\mathbb{N}$>> = $\iGoedel$ is mysteriously unprovable. If $\F$ then goes on to prove $\neg \iGoedel$ = <<$P(P)$ halts>>, this will be a false promise. Yet we're unable to catch it red-handed.

And if you thought that was neat, just wait until you see Rosser's trick. As you may recall, this one worked by patching the asymmetry between the two cases in Gödel's original argument. Let's see what happens if we try the same here.

\begin{theorem}[First Incompleteness Theorem---Rosser's Version By Computation]
Let $\F$ be a consistent Turing complete formal system. Then $\F$ is incomplete.
\end{theorem}

\begin{proof}\vspace{-\topsep}
Consider the following algorithm:
\begin{algo}
\Function{$B(A)$}{}\footnote{I picked the names $P$ and $B$ because there's a slight visual resemblance between the letters and the programs they denote. (Also this is why I placed the ``doesn't halt'' case first.)}
  \For{$x\in$ all possible strings}
    \If{$x$ proves <<$A(A)$ doesn't halt>>}
      \State halt
    \EndIf
    \If{$x$ proves <<$A(A)$ halts>>}
      \State loop
    \EndIf
  \EndFor
\EndFunction
\end{algo}
In a sense we've seen this algorithm before. Just take $H$ as it's defined in the first proof of Theorem \ref{original-first-by-computation}, substitute it into $\overline{H}$, and simplify. It's $\F$'s most valiant attempt to implement $\overline{H}$. This time our undecidable sentence is (you guessed it) $\iRosser$ = <<$B(B)$ doesn't halt>>.
\begin{description}
\item[Suppose $\iRosser$ is provable:]
$\Rightarrow$ There exists some shortest proof $p$ of $\iRosser$.
\begin{description}
\item[Either there is a proof of $\neg \iRosser$ that's shorter than $p$:]
\begin{itemize}
\item[]
\item $\F$ is inconsistent. \lightning\footnote{I might add that in this case $B(B)$ never reaches $p$ but finds the shorter proof first. We could dwell on this. But we already have an inconsistency which is all we need.}
\end{itemize}
\item[Or there isn't:]
\begin{itemize}
\item[]
\item $B(B)$ will go through all strings shorter than $p$ and determine that none of them are a proof of $\neg \iRosser$.
\item We can write out this computation inside $\F$ and point out that $B(B)$ will then find $p$ and enter a terminal state.
\item This allows us to construct a new (longer) proof of <<$B(B)$ halts>> = $\neg \iRosser$.
\item $\F$ is inconsistent. \lightning
\end{itemize}
\end{description}
\item[Suppose $\neg \iRosser$ is provable:]
$\Rightarrow$ There exists some shortest proof $p$ of $\neg \iRosser$.
\begin{description}
\item[Either there is a proof of $\iRosser$ that's shorter than $p$:]
\begin{itemize}
\item[]
\item $\F$ is inconsistent. \lightning
\end{itemize}
\item[Or there isn't:]
\begin{itemize}
\item[]
\item $B(B)$ will go through all strings shorter than $p$ and determine that none of them are a proof of $\iRosser$.
\item We can write out this computation inside $\F$ and point out that $B(B)$ will then find $p$ and enter an infinite loop.\footnote{Proving that a program runs forever raises the specter ($\mathghost$) of sub-inconsistency. But rest assured. Any formal system capable of reasoning about algorithms can prove that ``\textbf{while true} do nothing'' runs forever. Otherwise it ought to be ashamed of itself.}
\item This allows us to construct a new (longer) proof of <<$B(B)$ doesn't halt>> = $\iRosser$.
\item $\F$ is inconsistent. \lightning\qedhere
\end{itemize}
\end{description}
\end{description}
\end{proof}

I feel like I'm watching a reboot (in more than one sense). We can even recognize the phenomenon of finding either shorter or longer disproofs and proofs; pure satisfaction.

Alright, time for the grand finale.

\begin{theorem}[Second Incompleteness Theorem---By Computation]
One last time, let $\F$ be a consistent Turing complete formal system. Then $\F$ can't prove its own consistency.
\end{theorem}

We could continue using $B$ for the Second Incompleteness Theorem. But that would be ugly. Gödel didn't need Rosser's help for his proof. Neither shall we.

\begin{proof}
Let's assume for the purpose of contradiction that $\F$ can prove its own consistency
\begin{itemize}[topsep=0.5em]
\item The first half of our new proof for the First Incompleteness Theorem can be read as a proof by contradiction that if $\F$ is consistent, then $P(P)$ doesn't halt.
\item Using the fact that <<$P(P)$ doesn't halt>> = $\iGoedel$ this can be stated more concisely as ``$\F$ is consistent $\implies$ $\iGoedel$''.
\item Now we simply observe (with some faith) that there wasn't actually anything in the reasoning we used to prove this that we couldn't carry out just as well in $\F$. So <<$\F$ is consistent $\implies$ $\iGoedel$>> is actually a theorem in $\F$.
\item We're assuming that $\F$ can prove <<$\F$ is consistent>>. So by modus ponens $\F$ can also prove $\iGoedel$.
\item But then $P(P)$ will eventually find such a proof and halt. \lightning\qedhere
\end{itemize}
\end{proof}

\section{Taking Off The Glasses}

At this point I can't even tell the difference between the classical approach and the algorithmic retelling anymore. Can you? All these proofs are exactly the same as before! In fact, the relation is so strong that I was able to write the latter ones by copy-pasting wholesale and then tinkering a little. Just line them up and see for yourself. Every single line matches. Clearly, this can't be coincidence. There must be a deeper reason for this overabundance of similarities. And there is. Perhaps you've already caught on long ago. Or perhaps Clark Kent's glasses fooled you too. Drum whirl please.

Recall once more what $\iGoedel$ says. It says ``$P(P)$ doesn't halt''. But $P(P)$ isn't just executing some simple infinite loop. It's looking for something. It's scouring the proofs in $\F$ for a proof of <<$P(P)$ doesn't halt>> = $\iGoedel$. $\iGoedel$'s claim that $P(P)$ will never find a proof of $\iGoedel$ is echoing $G$'s claim that the original computers, \textit{you and me}, will never find a proof of $G$. In other words, $\iGoedel$ is also proclaiming to the world:
\begin{center}
``I am not provable!''
\end{center}
Exercise: Take off $\iRosser$'s glasses as well.

This means I could've organized this presentation very differently. I could've just made these observations right away and then said: ``Now reread the first section. \scalebox{0.9}{$\square$}''. I've been wasting your time!

\section{Refactoring The Incompleteness Theorems}

Gödel was programming the integers. That is that. And he didn't even know it at the time; truly impressive. I think this feat deserves to be encapsulated with its own theorem. Perhaps we can come together and save posterity from the grim fate of also only realizing this much later in their lives. It's for the kids---and for the horses.

\begin{theorem}
Elementary arithmetic is a Turing complete domain.
\end{theorem}

I'll have to leave the exact formulation of this theorem to smarter people. After all, I did recently score in the bottom 7\% of the GRE's analytical writing section (come on!). The upshot is that this allows us to easily recover the incompleteness theorems for elementary arithmetic from a more general insight:

\begin{theorem}[Turing Incompleteness Theorem]
Any consistent formal system $\F$ capable of reasoning about a Turing complete domain can reason about algorithms. That makes $\F$ a consistent Turing complete formal system. Thereby, the incompleteness theorems apply to it. In slogan form: 
\begin{center}
If it's complete, then it's incomplete.
\end{center}
\end{theorem}

\begin{proof}\vspace{-\topsep}
I think the idea is best conveyed with an example. Take the famously Turing complete Game of Life. For any algorithm one can create a configuration in the Game of Life which mimics that algorithm. Now suppose somebody came up with a formal system $\F$ to reason about the Game of Life. This formal system might be able to prove assertions such as: <<Starting with a configuration of the form [\dots] square $(0, 0)$ will never be occupied>>.

The incompleteness theorems are now immediate. All we have to do is to encode the algorithms $P$ and $B$ into a Game of Life configuration and reiterate the previous arguments. Conclusion: $\F$ will never be able to decide whether the encoding of $B$ ``halts''. And it won't be able to prove its own consistency.
\end{proof}

As a consequence computational undecidability and logical undecidability go hand in hand. And elementary arithmetic is just one further entry in the long list containing the Game of Life, \scalebox{0.95}{Fractran}, \scalebox{0.9}{Post Tag Systems}, \scalebox{0.85}{Magic the Gathering}, \scalebox{0.8}{the human brain}, \scalebox{0.75}{Rule 110}, \scalebox{0.7}{the Lambda Calculus}, \scalebox{0.65}{musical notation}, \scalebox{0.6}{\dots}\\[2em]
\footnotesize{(Thank you so much for your attention! The cake I promised in the abstract is on its way. Please stand by.) $\mathghost$}

\section*{Closing Credits}

I'd like to thank Scott Aaronson for his feedback as well as for being gracious enough to post this essay on his blog. Additionally, Cristopher Moore has provided many helpful comments. Finally, I strongly urge you to read Cristopher Moore and Stephan Mertens' book \textit{The Nature of Computation}. It's what made me fall in love with computer science.

\vfill\eject

\end{document}